\documentclass[11pt]{article}%
\DeclareMathAlphabet{\mathbbold}{U}{bbold}{m}{n}
\usepackage{amsmath,amsfonts,amssymb}
\usepackage{amsthm}
\usepackage{graphicx}
\usepackage{fullpage}
\usepackage[pagebackref]{hyperref}
\hypersetup{
    colorlinks=true, 
    linkcolor=blue, 
    citecolor=blue, 
    urlcolor=blue 
}
\usepackage{microtype}
\usepackage{thm-restate}
\usepackage[capitalize,nameinlink]{cleveref} 

\usepackage{bigfoot} 

\usepackage{tikz}
\usepackage{pgfplots}
\pgfplotsset{compat=1.13}
\usepackage{wrapfig}
\usepackage{enumerate}

\renewcommand{\backref}[1]{}

\renewcommand{\backrefalt}[4]{%
\ifcase #1 %
\or
[p.\ #2]%
\else
[pp.\ #2]%
\fi}

\makeatletter
\newcommand{\para}{%
 \@startsection{paragraph}{4}%
 {\z@}{2ex \@plus 3.3ex \@minus .2ex}{-1em}%
 {\normalfont\normalsize\bfseries}%
}
\makeatother

\allowdisplaybreaks

\providecommand{\U}[1]{\protect\rule{.1in}{.1in}}
\newtheorem{theorem}{Theorem}

\newtheorem{corollary}[theorem]{Corollary}

\newtheorem{lemma}[theorem]{Lemma}

\newtheorem{problem}[theorem]{Problem}

\let\oldproofname=\proofname
\renewcommand{\proofname}{\rm\bf{\oldproofname}}

\newcommand{\E}{\mathop{\mathbb{E}}}

\newcommand{\AND}{\mathsf{AND}} 
\newcommand{\OR}{\mathsf{OR}}
\newcommand{\NOR}{\mathsf{NOR}}

\newcommand{\Reals}{\mathbb{R}}

\renewcommand{\deg}{\mathrm{deg}}
\newcommand{\adeg}{\widetilde{\mathrm{deg}}}

\newcommand{\mathify}[1]{\ifmmode{#1}\else\mbox{$#1$}\fi}

\begin{document}

\title{\bfseries Lower Bounding the AND-OR Tree via Symmetrization}

\author{William Kretschmer\thanks{University of Texas at Austin. \ Email:
\texttt{kretsch@cs.utexas.edu}. \ Supported by a Simons Investigator Award.}}

\date{}
\maketitle

\begin{abstract}
We prove a simple, nearly tight lower bound on the approximate degree of the two-level $\AND$-$\OR$ tree using symmetrization arguments. Specifically, we show that $\adeg(\AND_m \circ \OR_n) = \widetilde{\Omega}(\sqrt{mn})$. We prove this lower bound via reduction to the $\OR$ function through a series of symmetrization steps, in contrast to most other proofs that involve formulating approximate degree as a linear program \cite{bt,sherstov,bbgk}. Our proof also demonstrates the power of a symmetrization technique involving Laurent polynomials (polynomials with negative exponents) that was previously introduced by Aaronson, Kothari, Kretschmer, and Thaler \cite{akkt}.
\end{abstract}



\section{Introduction}
\subsection{History of the AND-OR Tree}
The two-level $\AND$-$\OR$ tree has played an important role in the study of quantum query complexity. Given a set of $mn$ inputs over $\{0,1\}$, the problem is to compute the function $\AND_m \circ \OR_n = \mathop{\land}_{i=1}^m \mathop{\lor}_{j=1}^n x_{i,j}$. In the query model (see \cite{bd}), we assume access to an oracle that on input $(i,j)$ returns the bit $x_{i,j}$, and our goal is to compute $\AND_m \circ \OR_n(x_{1,1},\ldots,x_{m,n})$ in as few queries as possible.

One can show without too much difficulty that $\Theta(mn)$ queries are necessary and sufficient for any bounded-error randomized classical algorithm that computes $\AND_m \circ \OR_n$, by a standard adversary argument \cite{sw,ambainis}. But in the quantum query model, we can make queries on superpositions of inputs, and so the same lower bounds do not hold. In fact, one can do better in the quantum setting by using Grover's algorithm \cite{grover}, which with high probability finds a marked item in a list of $n$ items using just $O(\sqrt{n})$ queries. By applying Grover's algorithm recursively, along with an error reduction step on the inner subroutine, one obtains a quantum algorithm that makes $O(\sqrt{mn}\log m)$ queries \cite{bcw,hmd}. More nontrivially, H\o yer, Mosca, and de Wolf \cite{hmd} showed that by interleaving error-reduction with amplitude amplification, one can remove the extra log factor, giving a quantum algorithm that makes just $O(\sqrt{mn})$ queries.

The problem of lower bounding the quantum query complexity of the two-level $\AND$-$\OR$ tree proved to be much more challenging. Many of the early lower bounds on quantum query complexity were proved using the polynomial method of Beals et al. \cite{bbcmd}, which established a connection between quantum query complexity and approximate degree. The \textit{approximate degree} of a Boolean function $f: \{0,1\}^n \to \{0,1\}$, denoted $\adeg(f)$, is defined as the least degree of a polynomial $p(x_1,\ldots,x_n)$ over the reals that pointwise approximates $f$. We focus on constant factor approximations: that is, we require that for all $X \in \{0,1\}^n$, $p(X) \in [0, \alpha]$ whenever $f(X) = 0$, and $p(X) \in [\beta, 1]$ whenever $f(X) = 1$, for some constants $0 < \alpha < \beta < 1$. Note that the choices of $\alpha$ and $\beta$ are arbitrary, in the sense that varying them changes the approximate degree by at most a constant multiplicative factor. Beals et al. \cite{bbcmd} showed that for any quantum query algorithm that makes $T$ queries to a string $x \in \{0,1\}^n$, the acceptance probability of the algorithm can be expressed as a real polynomial of degree at most $2T$ in the bits of $x$. In particular, if the quantum query algorithm computes some function $f$ with bounded error, then this establishes that $\adeg(f)$ asymptotically lower bounds the number of quantum queries needed to compute $f$.

The $\AND$-$\OR$ tree initially resisted attempts at lower bounds via approximate degree. At the time, one of the only known tools for lower bounding approximate degree was \textit{symmetrization}, a technique that involves transforming a (typically symmetric) multivariate polynomial into a univariate polynomial.\footnote{Strictly speaking, the term ``symmetrization'' originally referred to the process of averaging a multivariate polynomial over permutations of its inputs, giving rise to a symmetric polynomial \cite{minsky}. We follow the convention of more recent work (e.g.\ \cite{sherstov2,bt2,akkt}), and use use ``symmetrization'' more generally to mean any process of transforming a multivariate polynomial to a univariate polynomial in a symmetry-exploiting, degree non-increasing way.} One typically appeals to classical results in approximation theory to analyze the resulting univariate polynomial. As an example, Paturi \cite{paturi} used symmetrization arguments to tightly characterize the approximate degree of all symmetric Boolean functions (i.e.\ functions that only depend on the Hamming weight of the input). While proofs by symmetrization are usually straightforward and easy to understand, for more complicated functions, symmetrization appears to have limited power in proving lower bounds. Indeed, symmetrization is inherently a lossy technique: a univariate polynomial can only capture part of the behavior of a multivariate polynomial.

Motivated by this difficulty, Ambainis \cite{ambainis} introduced a quantum analogue of the classical adversary method, and used it to prove the first tight lower bound of $\Omega(\sqrt{mn})$ for the quantum query complexity of $\AND_m \circ \OR_n$.\footnote{Note that Ambainis' $\Omega(\sqrt{mn})$ lower bound \cite{ambainis} predated the tight $O(\sqrt{mn})$ upper bound \cite{hmd}, but the $O(\sqrt{mn}\log m)$ upper bound was known at the time.} Still, the $\AND$-$\OR$ tree was the simplest Boolean function for which no tight approximate degree lower bound was known. Aaronson \cite{aaronson} even re-posed the question of resolving $\adeg(\AND_m \circ \OR_n)$ as a challenge problem for developing techniques beyond symmetrization.

After a series of successively tighter lower bounds on the $\AND$-$\OR$ tree (see \Cref{tab:history}), the approximate degree of $\AND_m \circ \OR_n$ was finally resolved in 2013, independently by Sherstov \cite{sherstov} and Bun and Thaler \cite{bt}. They showed that $\adeg(\AND_m \circ \OR_n) = \Omega(\sqrt{mn})$, answering a question that had been open for nearly two decades. Both lower bound proofs used the method of ``dual witnesses'' (or ``dual polynomials''), which involves formulating approximate degree as a linear program, then exhibiting a solution to the dual linear program to prove a lower bound on approximate degree. The dual witness method has the advantage that it can theoretically prove tight lower bounds for every Boolean function, and for general Boolean functions, the technique appears unaviodable. Indeed, many of the recent advances in polynomial approximation lower bounds seem to require such machinery \cite{BKT18,SW19,akkt}. Nevertheless, dual witness proofs tend to be more complicated, and finding and verifying explicit dual witnesses can be difficult (for example, notice in \Cref{tab:history} that it took 4 years to improve the lower bound between \cite{sherstov2} and \cite{bt,sherstov}!). To this date, essentially all tight lower bounds on $\adeg(\AND_m \circ \OR_n)$ rely on the dual formulation of approximate degree in some capacity.\footnote{\label{foot:bbgk}Recently, Ben-David, Bouland, Garg, and Kothari \cite{bbgk} proved that $\adeg(\AND_m \circ f) = \Omega\left(\sqrt{m}\;\adeg(f)\right)$ for any total Boolean function $f$. By a reduction involving polynomials derived from quantum algorithms, they show that $\adeg(\mathsf{XOR}_m \circ f) = O\left(\adeg(\AND_m \circ f)\sqrt{m}\right)$. They then appeal to the known result \cite{sherstov3} that $\adeg(\mathsf{XOR}_m \circ f) = \Omega\left(m \; \adeg(f)\right)$. Though they never construct an explicit dual witness, the lower bound on $\adeg(\mathsf{XOR}_m \circ f)$ relies on dual witnesses in an essential way.}

\begin{table}
\centering
\caption{\label{tab:history}History of lower bounds for $\adeg(\AND_n \circ \OR_n)$.}
\begin{tabular}{lll}
\textbf{Bound} & \textbf{Primary Technique} & \textbf{Reference}\\
\hline
\hline
$\Omega(\sqrt{n})$ & Symmetrization & Nisan and Szegedy \cite{ns}\\
$\Omega(\sqrt{n \log n})$ & Symmetrization & Shi \cite{shi}\\
$\Omega(n^{2/3})$ & Other & Ambainis \cite{ambainis2}\\
$\Omega(n^{3/4})$ & Dual witnesses & Sherstov \cite{sherstov2}\\
$\Omega(n)$ & Dual witnesses & Bun and Thaler \cite{bt}; Sherstov \cite{sherstov}\\
$\Omega(n)$ & Other (See \Cref{foot:bbgk}) & Ben-David, Bouland, Garg, and Kothari \cite{bbgk}\\
$\Omega\left(\frac{n}{\log n}\right)$ & Symmetrization & This paper
\end{tabular}
\end{table}

\subsection{Our Contribution: a Simpler Lower Bound}
In this work, we prove a nearly tight lower bound on the $\AND$-$\OR$ tree using very different techniques. Rather than going through the dual formulation of approximate degree directly, we lower bound the $\AND$-$\OR$ tree via reduction to the $\OR$ function through a series of symmetrization steps. Thus, our proof technique more closely mirrors some of the oldest lower bound proofs for symmetric Boolean functions \cite{minsky, ns, paturi}. Our proof is completely self-contained, with the exception of the lower bound on the $\OR$ function (which is easily proved via the classical Markov brothers' inequality \cite{ns}). Ultimately, we show that $\adeg(\AND_m \circ \OR_n) = \Omega\left(\frac{\sqrt{mn}}{\log m}\right)$, which is tight up to the log factor.

Crucially, our proof relies on a technique due to Aaronson, Kothari, Kretschmer, and Thaler \cite{akkt}: we use \textit{Laurent polynomials} (polynomials that can have both positive and negative exponents) and an associated symmetrization (\Cref{lem:laurent}) that reduces bivariate polynomials to univariate polynomials. In fact, our proof outline is very similar to the lower bound in \cite{akkt} on the one-sided approximate degree of $\mathsf{AND_2} \circ \mathsf{ApxCount}_{N,w}$. At a high level, our proof begins with what is essentially a lower bound on the degree of a ``robust'', partially symmetrized polynomial that approximates $\AND_m \circ \OR_n$. In particular, we show a generalization of the following:

\begin{theorem}[Informal]
Suppose that $p(x_1,\ldots,x_m)$ is a polynomial with the property that for all $(x_1,\ldots,x_m) \in [0, n]^m$:
\begin{enumerate}
\item If $x_i \le \frac{1}{3}$ for some $i$, then $0 \le p(x_1,\ldots,x_m) \le \frac{1}{3}$.
\item If $x_i \ge \frac{2}{3}$ for all $i$, then $\frac{2}{3} \le p(x_1,\ldots,x_m) \le 1$.
\end{enumerate}
Then $\deg(p) = \Omega(\sqrt{mn})$.
\end{theorem}

To give intuition, the variables $x_i$ roughly correspond to the Hamming weight of the inputs to each $\OR_n$ gate. Indeed, any polynomial that satisfies the statement of the theorem can be turned into one that approximates $\AND_m \circ \OR_n$ by letting $x_i$ equal the sum of the $\{0,1\}$ inputs to the $i$th $\OR_n$ gate. However, the polynomial is also required to be ``robust'' in the sense the polynomial must behave similarly when $x_i$ is not an integer.

The proof of \Cref{thm:L} works as follows: we group the $m$ variables into $\frac{m}{2}$ pairs and apply the Laurent polynomial symmetrization to each pair. We argue that this has the effect of ``switching'' the role of $\AND$ and $\OR$, in the sense that the resulting polynomial (in $\frac{m}{2}$ variables) looks like a partially symmetrized polynomial that approximates $\mathsf{NOR}_{m/2} \circ \OR_{\Theta(n)} = \mathsf{NOR}_{\Theta(mn)}$, which has approximate degree $\Omega(\sqrt{mn})$ by known lower bounds \cite{ns}.

We then show (\Cref{thm:robust}) that starting with a polynomial that approximates $\AND_m \circ \OR_n$, we can ``robustly symmetrize'' to construct a polynomial of the same degree that behaves like the one in the statement of \Cref{thm:L}, at the cost of a $\log m$ factor in the lower bound on the degree of the polynomial. This polynomial is obtained by applying the ``erase-all-subscripts'' symmetrization (\Cref{lem:subscripts}) to the variables corresponding to each $\OR_n$ gate, producing a polynomial in $m$ variables. This immediately implies (\Cref{cor:main}) that any polynomial that approximates $\AND_m \circ \OR_n$ has degree $\Omega\left(\frac{\sqrt{mn}}{\log m}\right)$.

\section{Preliminaries}
We use $[n]$ to denote the set $\{1,2, \ldots,n\}$, and use $\log$ to denote the natural logarithm. We will need the following two symmetrization lemmas, which were both introduced in these forms in \cite{akkt}. However, \Cref{lem:subscripts} is also a folklore result that previously appeared e.g.\ in \cite{shi} under the name ``linearization'':

\begin{lemma}[Erase-all-subscripts symmetrization \cite{shi, akkt}]\label{lem:subscripts}
Let $p(x_1,\ldots,x_n)$ be a real multilinear polynomial, and for any real number $\rho \in [0, 1]$, let $B\left(n, \rho\right)$ denote the distribution over $\{0,1\}^n$ wherein each coordinate is selected independently to be $1$ with probability $\rho$. Then there exists a real polynomial $q$ with $\deg(q) \le \deg(p)$ such that for all $\mu \in [0,n]$:
$$q(\mu) = \E_{(x_1,\ldots,x_n) \sim B\left(n, \frac{\mu}{n}\right)}[p(x_1,\ldots,x_n)].$$
\end{lemma}

\begin{proof}
Write:
$$q(\mu) = p\bigg(\underbrace{\frac{\mu}{n},\frac{\mu}{n},\ldots,\frac{\mu}{n}}_{n \text{ times}}\bigg).$$
Then the lemma follows from linearity of expectation, because $p$ is assumed to be multilinear.
\end{proof}

\begin{lemma}[Laurent polynomial symmetrization \cite{akkt}]\label{lem:laurent}
Let $p(x_1,x_2)$ be a real polynomial that is symmetric (i.e.\ $p(x_1,x_2) = p(x_2,x_1)$ for all $x_1, x_2$). Then there exists a real polynomial $q(t)$ with $\deg(q) \le \deg(p)$ such that for all real $s$, $q(s + 1/s) = p(s, 1/s)$.
\end{lemma}

\begin{proof}
Write $p(s, 1/s)$ as a Laurent polynomial $\ell(s)$ (a polynomial in $s$ and $1/s$). Because $p$ is symmetric, we have that $\ell(s) = p(s, 1/s) = p(1/s, s) = \ell(1/s)$. This implies that the coefficients of the $s^i$ and $s^{-i}$ terms of $\ell(s)$ are equal for all $i$, as otherwise $\ell(s) - \ell(1/s)$ would not be identically zero. Write $\ell(s) = \sum_{i=0}^d a_i \cdot (s^i + s^{-i})$ for some coefficients $a_i$, where $d \le \deg(p)$. Then, it suffices to show that $s^i + s^{-i}$ can be expressed as a real polynomial of degree $i$ in $s + 1/s$ for all $0 \le i \le d$.

We prove by induction on $i$. The case $i = 0$ is a constant polynomial. For $i > 0$, observe that $(s + 1/s)^i = s^i + s^{-i} + r(s)$, where $r(s)$ is some real Laurent polynomial of degree $i - 1$ satisfying $r(s) = r(1/s)$. In particular, we have that $s^i + s^{-i} = (s + 1/s)^i - r(s)$, where the right side of the equation is a real polynomial of degree $i$ in $s + 1/s$ (by the induction assumption).
\end{proof}

We remark that \Cref{lem:laurent} can also be viewed as a consequence of the fundamental theorem of symmetric polynomials,\footnote{Indeed, our proof even mirrors the standard proof of the fundamental theorem of symmetric polynomials.} which states that every symmetric polynomial in $n$ variables can be expressed uniquely as a polynomial in the elementary symmetric polynomials in $n$ variables. In 2 variables, the elementary symmetric polynomials are $x+y$ and $xy$. So, if $p(x, y) = q(x + y, xy)$, then restricting $p$ to the set $\{(s, 1/s): s \in \Reals\} = \{(x, y): xy = 1\}$ and writing in terms of $s$ corresponds to taking $q(s + 1/s, 1)$, which is just a polynomial in $s + 1/s$. (Of course, one would also have to show that this transformation can be applied in a degree-preserving way, as the proof of \Cref{lem:laurent} does).

Note also that \Cref{lem:subscripts} and \Cref{lem:laurent} can be applied to polynomials with additional variables, because of the isomorphism between the polynomial rings $\mathbb{R}[x_1,\ldots,x_n]$ and $\mathbb{R}[x_1,\ldots,x_k][x_{k+1},\ldots,x_n]$. For example, the Laurent polynomial symmetrization can be applied more generally to any polynomial $p(x_1,x_2,\ldots,x_n)$ that is symmetric in $x_1$ and $x_2$ by rewriting $p$ as a sum of the form:
\[p(x_1,x_2,\ldots,x_n) = \sum_i f_i(x_1, x_2) \cdot g_i(x_3,\ldots,x_n)\]
where $\{f_i\}$ and $\{g_i\}$ are sets of polynomials and the $f_i$s are all symmetric. Then, symmetrizing each $f_i$ according to \Cref{lem:laurent} gives a polynomial $q(s + 1/s, x_3, \ldots, x_n) = p(s, 1/s, x_3, \ldots, x_n)$.

Finally, we note the tight characterization of the approximate degree of $\OR$ and $\AND$:

\begin{lemma}[\cite{ns}]\label{lem:adeg_or}
$\adeg(\OR_n) = \adeg(\AND_n) = \Theta(\sqrt{n})$.
\end{lemma}

\section{Main Result}
We begin with the following theorem, which essentially lower bounds the degree of a robust, partially symmetrized polynomial that approximates $\AND_m \circ \OR_n$. Note that the case with $a$ and $b$ constant and $m = 2$ was essentially proved in \cite{akkt}.

\begin{theorem}\label{thm:L}
Let $0 < \alpha < \beta < 1$ be arbitrary constants, and let $0 < a < b < n$ such that $\frac{b}{a} < \frac{n}{b}$ also holds. Let $m \ge 2$, and suppose that $p(x_1,\ldots,x_m)$ is a polynomial with the property that for all $(x_1,\ldots,x_m) \in [0, n]^m$:
\begin{enumerate}
\item If $x_i \le a$ for some $i$, then $0 \le p(x_1,\ldots,x_m) \le \alpha$.
\item If $x_i \ge b$ for all $i$, then $\beta \le p(x_1,\ldots,x_m) \le 1$.
\end{enumerate}
Then $\deg(p) = \Omega\left(\frac{n - b}{b - a}\sqrt{\frac{a m}{n}}\right)$.
\end{theorem}

We first give a more intuitive, high-level overview of the proof of \Cref{thm:L}. For a symmetric polynomial $p(x_1, x_2)$, the Laurent polynomial symmetrization yields a polynomial $q(t)$ that captures the behavior of $p$ restricted to a hyperbola in which $x_1$ and $x_2$ are inversely proportional. Larger $t$ correspond to points on the hyperbola that are further from the origin, and thus correspond to points where one of $x_1$ or $x_2$ is small and the other is large. Smaller $t$ correspond to points on the hyperbola where both $x_1$ and $x_2$ are reasonably large.

The polynomial $p(x_1, \ldots, x_m)$ that we start with has the property that if any input $x_i$ is \textit{small} (less than $a$), then the polynomial should be close to 0. Otherwise, if all inputs are \textit{large} (greater than $b$), then the polynomial should be close to 1. The key observation is that the Laurent polynomial symmetrization essentially reverses the role of small and large inputs. The idea is to group the inputs of $p$ into $m/2$ pairs and apply the Laurent polynomial symmetrization to each pair, resulting in a polynomial $q(t_1,\ldots,t_{m/2})$ in half as many variables. Now, when some input $t_i$ to $q$ is large, this corresponds to some pair of variables in $p$ where one is small and one is large. Conversely, when all $t_i$s are are small, this corresponds to the case where all inputs to $p$ are reasonably large. As a result, if any input $t_i$ to $q$ is \textit{large}, then $q$ should be close to $0$, and otherwise, if all inputs are \textit{small}, then $q$ should be close to 1---precisely the reverse of $p$. For additional clarification, a diagram of these steps is shown in \Cref{fig:L}.

It remains to lower bound the degree of $q$. We do so by observing that $q$ can be turned into a polynomial that approximates the $\NOR$ function, whose approximate degree we know from \Cref{lem:adeg_or}.

\begin{figure}
\centering
\begin{tikzpicture}[y=0.28cm, x=0.28cm]
	\pgfmathsetmacro{\N}{20};
	\pgfmathsetmacro{\w}{3};
	\pgfmathsetmacro{\tw}{6};
	\pgfmathsetmacro{\Np}{\N+1};

	\fill[gray,opacity=.4] (0,0) -- (0,\N) -- (\w,\N) -- (\w,\w) -- (\N,\w) -- (\N,0);
	\fill[gray,opacity=.1] (\tw,\tw) -- (\tw,\N) -- (\N,\N) -- (\N,\tw);

	\draw[->] (0,0) -- coordinate (x axis mid) (\Np,0);
    \draw[->] (0,0) -- coordinate (y axis mid) (0,\Np);
    	
    \draw (0,1pt) -- (0,-3pt) node[anchor=north] {0};
    \draw (\w,1pt) -- (\w,-3pt) node[anchor=north] {$a$\vphantom{$b$}};
    \draw (\tw,1pt) -- (\tw,-3pt) node[anchor=north] {$b$};
    \draw (\N,1pt) -- (\N,-3pt) node[anchor=north] {$n$\vphantom{$b$}};

    \draw (1pt,0) -- (-3pt,0) node[anchor=east] {0};
    \draw (1pt,\w) -- (-3pt,\w) node[anchor=east] {$a$};
    \draw (1pt,\tw) -- (-3pt,\tw) node[anchor=east] {$b$};
    \draw (1pt,\N) -- (-3pt,\N) node[anchor=east] {$n$};

	\node[below=0.8cm] at (x axis mid) {$x_1$};
	\node[left=0.8cm] at (y axis mid) {$x_2$};
	\node at (\N/2, \w/2) {$0 \le p(x_1, x_2) \le \alpha$};
	\node at (\w+\N/2, \w+\N/2) {$\beta \le p(x_1, x_2) \le 1$};
	
	\draw[<-,thick] (\tw,\tw) -- (\tw+2,\tw+2) node[anchor=south] {$s_1=1$};
	\node at (\tw,\tw) {\tiny\textbullet}; 	
	\draw[<-,thick] (\w*4,\w) -- (\w*4+2,\w+2) node[anchor=south] {$s_1=\frac{b}{a}$};
	\node at (\w*4,\w) {\tiny\textbullet};
	\draw[<-,thick] (\N,4*\w*\w/\N) -- (\N+2,2+4*\w*\w/\N) node[anchor=south] {$s_1=\frac{n}{b}$};
	\node at (\N,4*\w*\w/\N) {\tiny\textbullet};
	\draw[<-,thick] (\w,\w*4) -- (\w+2,\w*4+2) node[anchor=south] {$s_1=\frac{a}{b}$};
	\node at (\w,\w*4) {\tiny\textbullet};
	\draw[<-,thick] (4*\w*\w/\N,\N) -- (4*\w*\w/\N+2,\N+2) node[anchor=south] {$s_1=\frac{b}{n}$};
	\node at (4*\w*\w/\N,\N) {\tiny\textbullet};
	
	\draw[<-,thick] (0.68*\w,4/0.68*\w) -- (0.68*\w+4.5,4/0.68*\w+3.5) node[anchor=west] {\color{blue}$(x_1=bs_1,x_2=b/s_1)$};

     \draw[thick,domain=1:{\N/(2*\w)},smooth,variable=\t,blue] plot ({2*\w*\t},{2*\w/\t});
     \draw[thick,domain=1:{\N/(2*\w)},smooth,variable=\t,blue] plot ({2*\w/\t},{2*\w*\t});

\end{tikzpicture}
\begin{tikzpicture}[y=0.28cm, x=0.28cm]
	\pgfmathsetmacro{\N}{20};
	\pgfmathsetmacro{\w}{3};
	\pgfmathsetmacro{\tw}{6};
	\pgfmathsetmacro{\Np}{\N+1};

	\fill[gray,opacity=.4] (\tw,\tw) -- (\tw,\w) -- (\N,\w) -- (\N,\N) -- (\w,\N) -- (\w,\tw);

	\draw[->] (0,0) -- coordinate (x axis mid) (\Np,0);
    \draw[->] (0,0) -- coordinate (y axis mid) (0,\Np);
    	
    \draw (0,1pt) -- (0,-3pt) node[anchor=north] {0};
    \draw (\w,1pt) -- (\w,-3pt) node[anchor=north] {2};
    \draw (\tw,1pt) -- (\tw,-3pt) node[anchor=north] {$\frac{b}{a} + \frac{a}{b}$};
    \draw (\N,1pt) -- (\N,-3pt) node[anchor=north] {$\frac{n}{b} + \frac{b}{n}$};

    \draw (1pt,0) -- (-3pt,0) node[anchor=east] {0};
    \draw (1pt,\w) -- (-3pt,\w) node[anchor=east] {2};
    \draw (1pt,\tw) -- (-3pt,\tw) node[anchor=east] {$\frac{b}{a} + \frac{a}{b}$};
    \draw (1pt,\N) -- (-3pt,\N) node[anchor=east] {$\frac{n}{b} + \frac{b}{n}$};

	\node[below=0.8cm] at (x axis mid) {$t_1$};
	\node[left=0.8cm] at (y axis mid) {$t_2$};
	\node at (\N/2+\w/2, \N/2+\w/2) {$0 \le q(t_1,t_2) \le \alpha$};
	
	\node at (\w, \w) {\tiny\textbullet};	
	
	\draw[<-,thick] (\w,\w) -- (\tw,0.5*\w) node[anchor=west] {$\beta \le q(t_1,t_2) \le 1$};

\end{tikzpicture}
\caption{\label{fig:L}Left: $p(x_1,x_2)$ for the case $m=2$. Right: $q(t_1,t_2)$ for the case $m=4$.}
\end{figure}
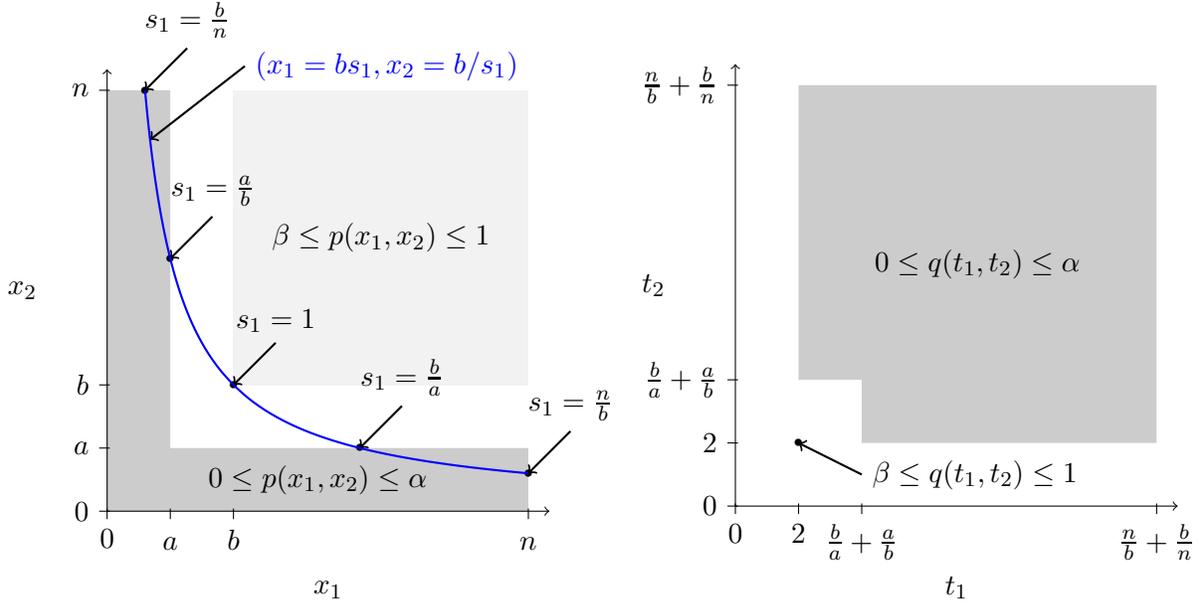

\begin{proof}
Assume that $m$ is even, since we can always set $x_m = b$ and consider the polynomial $p(x_1,\ldots,x_{m-1},b)$ on $m - 1$ variables instead. Assume without loss of generality that $p$ is symmetric in $x_1,\ldots,x_m$, because we can always replace $p$ by its average over all $m!$ permutations of the inputs. Group the variables into $m/2$ pairs $\{(x_{2i-1}, x_{2i}) : i \in [m/2]\}$ and apply the Laurent polynomial symmetrization (\Cref{lem:laurent}) to each pair to obtain a polynomial $q(t_1,\ldots,t_{m/2})$ in corresponding variables $\{t_i: i \in [m/2]\}$ with $\deg(q) \le \deg(p)$. We think of $t_i = s_i + 1/s_i$ as corresponding to the restriction $(x_{2i-1} = bs_i, x_{2i} = b/s_i)$ (note that this involves rescaling each input by $b$ in applying \Cref{lem:laurent}). Then we observe that for all $(t_1,\ldots,t_{m/2}) \in \left[2, \frac{n}{b} + \frac{b}{n}\right]^{m/2}$:
\begin{enumerate}
\item $0 \le q(t_1,\ldots,t_{m/2}) \le \alpha$ when $t_i \ge \frac{b}{a} + \frac{a}{b}$ for some $i$, as $t_i \ge \frac{b}{a} + \frac{a}{b}$ corresponds to either $s_i \ge \frac{b}{a}$ or $s_i \le \frac{a}{b}$, which corresponds to either $x_{2i} \le a$ or $x_{2i - 1} \le a$, respectively. (This is where we need the assumption $\frac{n}{b} > \frac{b}{a}$, as otherwise this case never holds).
\item $\beta \le q(t_1,\ldots,t_{m/2}) \le 1$ when $t_i = 2$ for all $i$, as $t_i = 2$ corresponds $s_i = 1$, which corresponds to $x_{2i - 1} = x_{2i} = b$.
\end{enumerate}
Perform an affine shift of $q$ with $\bar{t}_i = (t_i - 2)\frac{ab}{(b - a)^2}$ to obtain $\bar{q}(\bar{t}_1,\ldots,\bar{t}_{m/2}) = q(t_1,\ldots,t_{m/2})$. The reason for this choice is for convenience, so that the cutoffs on $t_i$ for the inequalities above become 1 and 0, respectively. It is easiest to see this by using the identity $\frac{b}{a} + \frac{a}{b} - 2 = \frac{(b-a)^2}{ab}$. Let $k = \left(\frac{n}{b} + \frac{b}{n} - 2\right)\frac{ab}{(b-a)^2} = \frac{a(n-b)^2}{n(b-a)^2}$, so that $t_i = \frac{n}{b} + \frac{b}{n}$ corresponds to $\bar{t}_i = k$. Note that $k \ge 1$ because the affine transformation is monotone, and $t_i = \frac{b}{a} + \frac{a}{b}$ corresponds to $\bar{t}_i = 1$. Then we observe that for all $(\bar{t}_1,\ldots,\bar{t}_{m/2}) \in \left[0,k\right]^{m/2}$:
\begin{enumerate}
\item $0 \le \bar{q}(\bar{t}_1,\ldots,\bar{t}_{m/2}) \le \alpha$ when $\bar{t}_i \ge 1$ for some $i$, as $\bar{t}_i \ge 1$ corresponds to $t_i \ge \frac{b}{a} + \frac{a}{b}$.
\item $\beta \le \bar{q}(\bar{t}_1,\ldots,\bar{t}_{m/2}) \le 1$ when $\bar{t}_i = 0$ for all $i$, as $\bar{t}_i = 0$ corresponds $t_i = 2$.
\end{enumerate}
Notice that $\bar{q}$ approximates a partially symmetrized $\NOR$ function. Now, we ``un-symmetrize'' $\bar{q}$. Let $\bar{t}_i = \bar{t}_{i,1} + \bar{t}_{i,2} + \ldots + \bar{t}_{i,\lfloor k \rfloor}$. Then $\bar{q}(\bar{t}_{1, 1} + \ldots + \bar{t}_{1,\lfloor k \rfloor}, \bar{t}_{2, 1} + \ldots + \bar{t}_{2,\lfloor k \rfloor}, \ldots, \bar{t}_{m/2,1} + \ldots \bar{t}_{m/2,\lfloor k \rfloor})$ approximates $\NOR_{m \lfloor k \rfloor/2}$ over the variables $(\bar{t}_{1,1},\ldots,\bar{t}_{m/2,\lfloor k \rfloor}) \in \{0,1\}^{m \lfloor k \rfloor/2}$. Since we know that $\adeg(\NOR_{m\lfloor k \rfloor/2}) = \adeg(\OR_{m\lfloor k \rfloor/2}) = \Omega(\sqrt{mk})$ (\Cref{lem:adeg_or}), and since this construction satisfies $\deg(\bar{q}) = \deg(q) \le \deg(p)$, we conclude that $\deg(p) = \Omega\left(\frac{n - b}{b - a}\sqrt{\frac{a m}{n}}\right)$.
\end{proof}

Next, we show that a polynomial that approximates $\AND_m \circ \OR_n$ can be ``robustly symmetrized'' like in the statement of \Cref{thm:L} with $a = O(1)$ and $b = O(\log m)$.

\begin{theorem}
\label{thm:robust}
Let $p(x_{1,1},\ldots,x_{m,n})$ be a polynomial in variables $\{x_{i,j} : (i, j) \in [m] \times [n]\}$ that $\frac{1}{3}$-approximates $\AND_m \circ \OR_n$, where $n > 2 \log m$ and $m \ge 10$. Specifically, we assume that $0 \le p(x_{1,1},\ldots,x_{m,n}) \le \frac{1}{3}$ on a $0$-instance, and $\frac{2}{3} \le p(x_{1,1},\ldots,x_{m,n}) \le 1$ on a $1$-instance. Then there exists a polynomial $q(x_1,\ldots,x_m)$ with $\deg(q) \le \deg(p)$ such that for all $(x_1,\ldots,x_m) \in [0, n]^m$:
\begin{enumerate}
\item If $x_i \le \frac{1}{6}$ for some $i$, then $0 \le q(x_1,\ldots,x_m) \le \frac{1}{2}$.
\item If $x_i \ge 2\log m$ for all $i$, then $\frac{3}{5} \le q(x_1,\ldots,x_m) \le 1$.
\end{enumerate}
\end{theorem}
\begin{proof}
Assume without loss of generality that $p$ is multilinear (because $x^2 = x$ over $\{0,1\}$), so that we can apply the erase-all-subscripts symmetrization (\Cref{lem:subscripts}) separately to the inputs of each $\OR$ gate. This erases all of the ``$j$'' subscripts, giving a polynomial $q(x_1,\ldots,x_m)$ with $\deg(q) \le \deg(p)$ such that:
$$q(x_1,\ldots,x_m) = \E_{(x_{i,1},\ldots,x_{i,n}) \sim B\left(n, \frac{x_i}{n}\right)}\left[ p(x_{1,1},\ldots,x_{m,n}) \right]$$
where $B(n, \rho)$ denotes the distribution over $\{0,1\}^n$ where each coordinate is selected from an independent Bernoulli distribution with probability $\rho$.

Suppose $x_i \in [2\log m, n]$ for all $i \in [m]$. Then we have that:
\begin{align}
\Pr_{(x_{i,1},\ldots,x_{i,n}) \sim B\left(n, \frac{x_i}{n}\right)} \left[ \AND_m \circ \OR_n (x_{1,1},\ldots,x_{m,n}) = 0\right]
&\le
\sum_{i=1}^m \Pr\left[(x_{i,1},\ldots,x_{i,n}) = 0^n\right] \label{eq:upper1}\\
&\le
m \cdot \left(1 - \frac{2 \log m}{n}\right)^n \label{eq:upper2}\\
&\le \frac{1}{m} \label{eq:upper3}
\end{align}
where \eqref{eq:upper1} follows from a union bound, \eqref{eq:upper2} follows by expanding each term, and \eqref{eq:upper3} follows from the exponential inequality.

On the other hand, suppose $x_i \in [0, n]$ for all $i \in [m]$, and that $x_{i^*} \in [0, \frac{1}{6}]$ for some $i^* \in [m]$. Then we have that:
\begin{align*}
\Pr_{(x_{i,1},\ldots,x_{i,n}) \sim B\left(n, \frac{x_i}{n}\right)} \left[ \AND_m \circ \OR_n (x_{1,1},\ldots,x_{m,n}) = 0\right]
&\ge
\Pr\left[(x_{i^*,1},\ldots,x_{i^*,n}) = 0^n\right]\\
&\ge
\left(1 - \frac{1}{6n}\right)^n\\
&\ge \frac{5}{6}
\end{align*}
by similar inequalities.

Since in general we can bound $q(x_1,\ldots,x_m)$ by
$$q(x_1,\ldots,x_m) \ge \frac{2}{3} \cdot \Pr_{(x_{i,1},\ldots,x_{i,n}) \sim B\left(n, \frac{x_i}{n}\right)}\left[ \AND_m \circ \OR_n (x_{1,1},\ldots,x_{m,n}) = 1\right]$$
and
$$q(x_1,\ldots,x_m) \le \frac{1}{3} \cdot \Pr\left[ \AND_m \circ \OR_n (x_{1,1},\ldots,x_{m,n}) = 0\right] + 1 \cdot \Pr\left[ \AND_m \circ \OR_n (x_{1,1},\ldots,x_{m,n}) = 1\right],$$
we can conclude that for all $(x_1,\ldots,x_m) \in [0, n]^m$:
\begin{enumerate}
\item $0 \le q(x_1,\ldots,x_m) \le \frac{1}{3}\cdot 1 + 1\cdot\frac{1}{6} = \frac{1}{2}$ when $x_i \le \frac{1}{6}$ for some $i$.
\item $\frac{3}{5} \le \frac{2}{3}(1 - \frac{1}{m}) \le q(x_1,\ldots,x_m) \le 1$ when $x_i \ge 2\log m$ for all $i$.\qedhere
\end{enumerate}
\end{proof}

Putting these two theorems together gives:

\begin{corollary}
\label{cor:main}
$\adeg(\AND_m \circ \OR_n) = \Omega\left(\frac{\sqrt{mn}}{\log m}\right)$.
\end{corollary}

\begin{proof}
If $n \le 24\log^2 m$, then the theorem holds trivially by the lower bound on $\adeg(\AND_m)$ (\Cref{lem:adeg_or}), since $\adeg(\AND_m \circ \OR_n) \ge \adeg(\AND_m) = \Omega(\sqrt{m})$. Likewise, if $m < 10$, then the theorem holds trivially by the lower bound on $\adeg(\OR_n)$. Otherwise, putting \Cref{thm:L} and \Cref{thm:robust} together gives $\adeg(\AND_m \circ \OR_n) = \Omega\left(\frac{n - 2\log m}{2 \log m - 1/6}\sqrt{\frac{m}{6n}}\right) = \Omega\left(\frac{\sqrt{mn}}{\log m}\right)$.
\end{proof}

\section{Discussion}
Though our lower bound on $\adeg(\AND_m \circ \OR_n)$ is not tight, it suggests a natural way to obtain a tighter lower bound via the same method: either tighten the ``robust symmetrization'' argument in \Cref{thm:robust} to eliminate the log factor (i.e.\ construct a polynomial of the same degree with $a = \Omega(1)$ and $b = O(1)$), or show that the lower bound in \Cref{thm:L} can be tightened by a $\log m$ factor. We conjecture that \Cref{thm:L} is tight, and that \Cref{thm:robust} can be improved.

We remark that nevertheless, it is an open problem to exhibit polynomials of minimal degree that satisfy the statement of \Cref{thm:L} with $a$ and $b$ constant, \textit{even in the case where $n=1$}! We highlight this below:

\begin{problem}\label{prob:robust}
What is the minimum degree of a polynomial $p(x_1,\ldots,x_m)$ with the property that for all $(x_1,\ldots,x_m) \in [0, 1]^m$:
\begin{enumerate}
\item If $x_i \le \frac{1}{3}$ for some $i$, then $0 \le p(x_1,\ldots,x_m) \le \frac{1}{3}$.
\item If $x_i \ge \frac{2}{3}$ for all $i$, then $\frac{2}{3} \le p(x_1,\ldots,x_m) \le 1$.
\end{enumerate}
In particular, do such polynomials exist of degree $O(\sqrt{m})$?
\end{problem}

While degree $\Omega(\sqrt{m})$ polynomials are clearly necessary (\Cref{lem:adeg_or}), the best upper bound we are aware of for \Cref{prob:robust} is $O\left(\sqrt{m}c^{\log^* m}\right)$ for some constant $c$, where $\log^*$ denotes the iterated logarithm. Such polynomials can be derived from a quantum algorithm for search on bounded-error inputs described in the introduction of \cite{hmd}, using the connection between quantum query algorithms and approximating polynomials \cite{bbcmd}.\footnote{Note that the main result of \cite{hmd} is of no use here, as it requires all inputs to be in $[0,\frac{1}{3}] \cup [\frac{2}{3}, 1]$, whereas case 1 of \Cref{prob:robust} allows inputs in $[\frac{1}{3}, \frac{2}{3}]$.}

On the other hand, we note that the conditions of \Cref{thm:L} are stronger than necessary. Upon closer observation, it suffices to have a polynomial that satisfies the following weaker conditions:
\begin{enumerate}
\item If $x_i \in [0, a] \cup [b, n]$ for all $i$ and $x_i \le a$ for some $i$, then $0 \le p(x_1,\ldots,x_m) \le \alpha$.
\item If $x_i = b$ for all $i$, then $\beta \le p(x_1,\ldots,x_m) \le 1$.
\end{enumerate}

And indeed, under these weaker conditions, the existence of such polynomials of degree $O(\sqrt{m})$ for $n = 1, a = \frac{1}{3}, b = \frac{2}{3}$ follows from the main result of H\o yer, Mosca, and de Wolf \cite{hmd}. Alternatively, the existence of such polynomials follows from Sherstov's result on making polynomials robust to noise \cite{sherstov4}.

Beyond the results that we (re-)proved, we wonder if the techniques introduced will find applications elsewhere. For example, we observed a connection between \Cref{lem:laurent} and the fundamental theorem of symmetric polynomials. Is the Laurent polynomial symmetrization a special case of a more general class of symmetrizations that involve analyzing symmetric polynomials in the basis of elementary symmetric polynomials? If so, do any of these other types of symmetrizations have applications in proving new lower bounds on approximate degree?

\subsection{Followup work}
After this work was published, we became aware of a recent work, due to Huang and Viola \cite{HV20}, that gives yet another lower bound proof for the approximate degree of the $\AND$-$\OR$ tree. The proof is based around the notion of \textit{$k$-wise indistinguishability}, which captures distributions over $n$-bit strings for which the marginal distributions on any $k$ bits are the same.

\section*{Acknowledgements}
This paper originated as a project in Scott Aaronson's Spring 2019 Quantum Complexity Theory course; I am grateful for his guidance. I thank Robin Kothari for bringing \Cref{prob:robust} and \cite{bbgk} to my attention. Thanks also to Justin Thaler for helpful discussions and feedback that were the inspiration for this paper.


\phantomsection\addcontentsline{toc}{section}{References}
\bibliographystyle{alphaurl}
\bibliography{AND-OR2}

\end{document}